\documentclass[a4paper,12pt]{article}
\pdfoutput=1

\usepackage{test}
\usepackage[colorlinks,citecolor=blue,urlcolor=magenta]{hyperref}

\usepackage[style=ext-authoryear-comp,
sorting=nyt,
dashed=false, 
maxcitenames=2, 
maxbibnames=99, 
uniquelist=false,
uniquename=false,
giveninits=true, 
natbib, 
date=year 
]{biblatex}

\AtBeginRefsection{\GenRefcontextData{sorting=ynt}}
\AtEveryCite{\localrefcontext[sorting=ynt]}

\DeclareFieldFormat{pages}{#1} 
\renewbibmacro{in:}{\ifentrytype{article}{}{\printtext{\bibstring{in}\intitlepunct}}} 

\DeclareFieldFormat[article,inbook,incollection,inproceedings,patent,thesis,unpublished]{titlecase:title}{\MakeSentenceCase*{#1}} 

\usepackage[inline,shortlabels]{enumitem}
\setlist[enumerate,1]{label=(\roman*)}

\usepackage[margin = 1.25in]{geometry}
\usepackage{setspace}
\title{Rational Bubbles Attached to Real Assets\thanks{This document is a draft of a chapter to be included in \emph{Handbook of Economic Bubbles and Manias} edited by Manuel Santos. We thank Ngoc-Sang Pham and Manuel Santos for comments.}}
\author{Tomohiro Hirano\thanks{Department of Economics, Royal Holloway, University of London, and Research Associate at the Center for Macroeconomics at the London School of Economics,  \href{mailto:tomohiro.hirano@rhul.ac.uk}{tomohiro.hirano@rhul.ac.uk}.} \and Alexis Akira Toda\thanks{Department of Economics, Emory University, \href{mailto:alexis.akira.toda@emory.edu}{alexis.akira.toda@emory.edu}.}}

\numberwithin{equation}{section}
\numberwithin{lem}{section}
\numberwithin{prop}{section}

\begin{document}

\maketitle

\begin{abstract}
A rational bubble is a situation in which the asset price exceeds its fundamental value defined by the present discounted value of dividends in a rational equilibrium model.  We discuss the recent development of the theory of rational bubbles attached to real assets, emphasizing the following three points.
\begin{enumerate*}
    \item There exist plausible economic models in which bubbles inevitably emerge in the sense that all equilibria are bubbly.
    \item Such models are necessarily nonstationary but their long-run behavior can be analyzed using the local stable manifold theorem.
    \item Bubbles attached to real assets can naturally and necessarily arise with economic development.
\end{enumerate*}
We illustrate these three points in various settings attesting that bubbles may necessarily emerge for aggregate stocks, land, and other assets.

\medskip

\noindent
\textbf{Keywords:} bubble, elasticity, nonstationarity, productivity, unbalanced growth.

\medskip

\noindent
\textbf{JEL codes:} D53, E44, G12, O16.
\end{abstract}

\section{Introduction}

An asset price bubble occurs when ``asset prices do not reflect fundamentals'' \citep{Stiglitz1990}, or, in other words, when the asset price ($P$) exceeds its fundamental value ($V$) defined as the present discounted value of dividends ($D$). Throughout financial history, notable bubbly episodes that readily come to mind include the Japanese real estate and stock market bubble of the late 1980s, the {U.S.} dot-com bubble of the late 1990s, and the {U.S.} housing bubble in the mid-2000s.\footnote{\citet[Appendix B]{Kindleberger2000} documents 38 bubbly episodes in the 1618--1998 period, an average of one episode every ten years.} Although history is replete with bubbly episodes, it is well known in macro-finance theory that it is notoriously difficult to generate asset price bubbles ($P>V$) in rational equilibrium models with real assets. By ``real assets'', we mean assets that pay positive dividends ($D>0$) such as stocks, land, and housing. In fact, in a seminal paper on asset price bubbles, \citet{SantosWoodford1997} proved the following bubble impossibility result: in a general equilibrium model with rational optimizing agents, if aggregate dividends comprise a non-negligible fraction of aggregate endowments, asset price bubbles cannot arise.\footnote{This result follows from Theorem 3.3 and Corollary 3.4 of \citet{SantosWoodford1997}. See \citet[\S3.4]{HiranoToda2024JME} for a simple illustration.} This bubble impossibility result implies that there is no basic macro-finance theory to think about bubbles attached to real assets in the first place, despite the fact that the term ``bubble'' frequently appears in the popular press or policy circles. Moreover, perhaps due to the result, in the academic world of macro-finance, it seems that there is a presupposition that asset prices must reflect fundamentals. Indeed, it is fair to say that these models are constructed so that the asset price is always equal to the fundamental value.

Due to this fundamental difficulty in attaching bubbles to dividend-paying assets, the rational bubble literature has almost exclusively focused on the so-called ``pure bubble'' model in which the asset pays no dividends, like fiat money.\footnote{\label{fn:RB}This literature starts with the seminal paper of \citet{Samuelson1958}; see \citet{HiranoToda2024JME} for a recent review. See \citet{BrunnermeierOehmke2013} for an introduction to bubbles including other approaches such as heterogeneous beliefs and asymmetric information. \citet{MartinVentura2018} review macroeconomic applications of rational bubble models.} However, pure bubble models are subject to several criticisms.
\begin{enumerate*}
    \item First, the assumption of zero dividends is unrealistic because most assets in the real world other than fiat money or cryptocurrency do pay dividends.
    \item Second, equilibria in pure bubble models are often indeterminate. As shown by \citet{Gale1973}, in pure bubble models, there exists a steady state in which the asset has a positive value, as well as a continuum of equilibria in which the asset value converges to zero.\footnote{This statement is often true but not always. See \citet{Scheinkman1980} and \citet{Santos1990} for counterexamples of indeterminacy, though they require strong assumptions (\eg, no endowment of future goods). \citet{HiranoToda2024EL} examine if their result also holds in production economies and prove that there exist a continuum of monetary equilibria.} This equilibrium indeterminacy makes model predictions non-robust.
    \item Third, with zero dividends, the price-dividend ratio is undefined, which makes it impossible to connect to the econometric literature on bubble detection that uses the price-dividend ratio \citep{PhillipsShiYu2015,PhillipsShi2018,PhillipsShi2020}.
\end{enumerate*}
These criticisms simply show that in describing bubbles attached to real
assets, pure bubble models face fundamental limitations for applications including policy and quantitative analyses \citep{Barlevy2018}. If models cannot be applied, it will be difficult for the literature to develop.\footnote{On this point, we thank Jos\'e Scheinkman and Nobuhiro Kiyotaki for pointing out the limitations of pure bubble models and teaching us how difficult and how valuable it is to prove the existence of rational bubbles attached to real assets with positive dividends in a modern macro-finance framework. Indeed, when one of the authors (Hirano) presented earlier papers on pure bubbles, the reaction from the general audience was harsh, some claiming that pure bubble models are useless in thinking about realistic bubbles attached to stocks, land, and housing due to these criticisms.}

In this chapter, we discuss the recent development of the theory of rational asset price bubbles attached to real assets. By doing so, we would like to bridge the gap between the popular press, policy circles, and academia. We also would like to convince researchers who believe that bubbles cannot possibly occur. In particular, we demonstrate that under some conditions, bubbles are not just possible but inevitable using widely applied macro-finance models.

We emphasize the following three points.
\begin{enumerate*}
    \item First, in \S\ref{sec:necessity} we explain the concept of the \emph{necessity} of bubbles proposed by the recent paper of \citet{HiranoToda2025JPE}. The idea is that, when the dividend growth rate of the asset exceeds the counterfactual autarky interest rate (the interest rate that would prevail in a counterfactual economy without the asset) but is below the economic growth rate, then bubbles necessarily emerge in equilibrium. We also discuss several concrete examples in \S\ref{sec:example}.
    \item Second, bubbles attached to real assets entail a world of nonstationarity, which requires analytical tools. To make the theory appealing to applied researchers, in \S\ref{sec:longrun} we explain how to apply the local stable manifold theorem (which is essentially linearization) to quantitatively study the long-run behavior of asset prices in such models.
    \item Third, we show that the emergence of bubbles and economic development are closely related. To illustrate this point, in \S\ref{sec:longrun} we present an overlapping generations model with a dividend-paying asset. We show that when the incomes of the young become sufficiently high relative to the incomes of the old, \ie, economic development, asset price bubbles become inevitable. Moreover, asset price volatility would be highest with a medium level of economic development. Finally, in \S\ref{sec:land} we present another overlapping generations model with two assets, stocks and land. There are two sectors, a capital-intensive sector (\eg, manufacturing) and a land-intensive sector (\eg, agriculture). In the capital-intensive sector, firms produce the output using capital and labor. Stock shares are issued backed by the returns generated by capital. In the land sector, land produces the output as dividends. Both firm stocks and land are traded as long-lived assets. We show that under certain conditions on the elasticity of substitution in the production function and the productivity growth rates, bubbles in aggregate stock and land prices necessarily emerge.
\end{enumerate*}

\paragraph{Cautionary warning}

We end the introduction with a cautionary warning. Because the term ``bubble'' is commonly used in the popular press, policy circles, and academia, a scientific discourse requires clarity on this term. Unfortunately, scientific language is not standardized, as ``[the] word ``bubble'' is widely used to mean very different things. Some people seem to mean any large movement in prices'' \citep[p.~404]{Cochrane2005}. In this chapter, we use the term ``bubble'' in the classical sense to be defined in \S\ref{sec:bubble}, \ie, a situation in which the asset price exceeds its fundamental value defined by the present discounted value of dividends even though agents are rational and have common beliefs or information. This approach is precisely what constitutes the ``rational bubble'' model pioneered by \citet{Samuelson1958}, \citet{Tirole1985}, \citet{Kocherlakota1992}, \citet{SantosWoodford1997}, and others (Footnote \ref{fn:RB}). However, we do not claim by any means that the ``rational bubble'' approach is the only possibility or is superior to other approaches. We acknowledge that there are various approaches to asset price bubbles including heterogeneous beliefs \citep{ScheinkmanXiong2003,FostelGeanakoplos2012Tranching}, asymmetric information \citep{AbreuBrunnermeier2003,Barlevy2014,AllenBarlevyGale2022,AwayaIwasakiWatanabe2022}, liquidity \citep{RocheteauWright2013,LagosRocheteauWright2017}, self-fulfilling expectations \citep{Krishnamurthy2003,MiaoWang2018}, and uncertainty and learning \citep{PastorVeronesi2009}, among others. We welcome diverse perspectives and hope that the science of bubbles will prosper.

In general, as \citet{HiranoToda2025EJW} emphasize, researchers should refrain from carelessly using the term ``bubble'' without a formal mathematical definition and proof. This misconception occurs in many other areas of economics in which numerical experiments are often confused with formal results without a clear understanding of the underlying conditions generating the outcome of these experiments. Hence, a ``cautionary warning'': readers should
always be aware of such pitfalls as authors are not always clear about the formal mathematical definition and proof. As discussed in \citet{BonaSantos1997}, numerical results may sometimes lead to valuable conjectures, but cannot be seen with the eyes of a theorem or other formal statements. These authors prescribe some useful guidelines as to how numerical experiments should be reported and interpreted.

\section{Rational bubbles as speculation}\label{sec:bubble}

The formal definition of rational bubbles was gien by \citet{SantosWoodford1997}. Here we follow the discussion in \citet{HiranoToda2025EJW}.

\subsection{Formal definitions}\label{subsec:definition}

Consider an infinite-horizon economy with a homogeneous good and time indexed by $t=0,1,\dotsc$. Let $\pi_t$ denote a state price deflator. For instance, in a deterministic economy, $\pi_t$ is the date-0 price of a zero-coupon bond with maturity $t$. Consider an asset with infinite maturity that pays dividend $D_t\ge 0$ and trades at ex-dividend price $P_t$, both in units of the time-$t$ good. Then the no-arbitrage asset pricing equation is given by
\begin{equation}
    \pi_t P_t = \E_t[\pi_{t+1}(P_{t+1}+D_{t+1})]. \label{eq:noarbitrage_pi}
\end{equation}
Solving this equation forward by repeated substitution (and applying the law of iterated expectations) yields
\begin{equation}
    \pi_tP_t=\E_t\sum_{s=t+1}^T \pi_sD_s+\E_t[\pi_TP_T]. \label{eq:P_iter}
\end{equation}
Because all terms are nonnegative, the sum in \eqref{eq:P_iter} from $s=t+1$ to $s=T$ is
\begin{enumerate*}
    \item increasing in $T$ and
    \item bounded above by $\pi_tP_t$,
\end{enumerate*}
so it converges almost surely as $T\to\infty$. Therefore the \emph{fundamental value} of the asset
\begin{equation}
    V_t\coloneqq \frac{1}{\pi_t}\E_t\sum_{s=t+1}^\infty \pi_sD_s \label{eq:Vt}
\end{equation}
is well-defined, and letting $T\to\infty$ in \eqref{eq:P_iter}, we obtain $P_t=V_t+B_t$, where we define the \emph{asset price bubble} as
\begin{equation}
    B_t\coloneqq \lim_{T\to\infty} \frac{1}{\pi_t}\E_t[\pi_TP_T]\ge 0.\label{eq:Bt}
\end{equation}
That is, an asset price bubble is equal to the difference between the market price of the asset and its fundamental value (\ie, the present value of dividends). By definition, there is no bubble at time $t$ if and only if the \emph{no-bubble condition}
\begin{equation}
    \lim_{T\to\infty} \E_t[\pi_TP_T]=0 \label{eq:TVC}
\end{equation}
holds. This is the mathematical formalization of the idea explained in \citet{Stiglitz1990}. Conditions at infinity like \eqref{eq:TVC} are often called \emph{transversality conditions} \citep{MagillQuinzii1994,MagillQuinzii1996,SantosWoodford1997,Montrucchio2004,GiglioMaggioriStroebel2016}. In our earlier papers \citep{HiranoToda2024JME,HiranoToda2025JPE}, we referred to \eqref{eq:TVC} as the  \emph{transversality condition for asset pricing} (to distinguish from the transversality condition for optimality in infinite-horizon optimal control problems; see \citet[Ch.~15]{TodaEME}). To prevent confusion, here we simply refer to \eqref{eq:TVC} as the no-bubble condition.

The economic meaning of the bubble component $B_t$ in \eqref{eq:Bt} is that it captures a speculative aspect, that is, agents buy the asset now for the purpose of resale in the future, rather than for the purpose of receiving dividends. When the no-bubble condition \eqref{eq:TVC} holds, the aspect of speculation becomes negligible and asset prices are determined only by factors that are backed in equilibrium, namely future dividends. On the other hand, if $\lim_{T\to\infty}\E_t[\pi_TP_T]>0$, equilibrium asset prices contain a speculative aspect backed by nothing and are strictly higher than the present discount value of the dividend stream.

\subsection{Bubble Characterization Lemma}\label{subsec:characterization}

To prove the existence of rational bubbles, we need to prove $P>V$, or equivalently, verify the violation of the no-bubble condition \eqref{eq:TVC}. For an asset that pays no dividends ($D=0$, pure bubble), because the fundamental value is necessarily zero, showing $P>0$ suffices. However, for dividend-paying assets, \ie, real assets such as stocks, land, and housing, the verification of the violation of the no-bubble condition is not easy because it is cumbersome to calculate the state price deflator $\pi_t$. Fortunately, in economies without aggregate uncertainty, there is a very simple characterization due to \citet[Proposition 7]{Montrucchio2004}.\footnote{\citet{Montrucchio2004} and \citet{CruzRambaud2013} consider the case with aggregate uncertainty but they focus on sufficient conditions for the nonexistence of bubbles.} The statement and proof below follows \citet[Lemma 1]{HiranoToda2025JPE}.

\begin{lem}[Bubble characterization]\label{lem:bubble}
In an economy without aggregate uncertainty, if $P_t>0$ for all $t$, the asset price exhibits a rational bubble if and only if $\sum_{t=1}^\infty D_t/P_t<\infty$.
\end{lem}

\begin{proof}
If the asset is risk-free, taking the unconditional expectations of \eqref{eq:noarbitrage_pi} and setting $q_t=\E[\pi_t]>0$ (which equals the date-0 price of a zero-coupon bond with maturity $t$), we obtain
\begin{equation}
    q_tP_t = q_{t+1}(P_{t+1}+D_{t+1}). \label{eq:noarbitrage}
\end{equation}
Then by the same argument as in \S\ref{subsec:definition} and using $q_0=1$, we obtain
\begin{equation}
    P_0=\sum_{t=1}^T q_tD_t+q_TP_T, \label{eq:P0_iter}
\end{equation}
and there is no bubble if the no-bubble condition $\lim_{T\to\infty}q_TP_T=0$ holds.

Changing $t$ to $t-1$ in the no-arbitrage condition \eqref{eq:noarbitrage} and dividing both sides by $q_tP_t>0$, we obtain $q_{t-1}P_{t-1}/q_tP_t=1+D_t/P_t$. Multiplying from $t=1$ to $t=T$, expanding terms, and using $1+x\le \e^x$, we obtain
\begin{equation*}
    1+\sum_{t=1}^T\frac{D_t}{P_t}\le \frac{q_0P_0}{q_TP_T}=\prod_{t=1}^T\left(1+\frac{D_t}{P_t}\right)\le \exp\left(\sum_{t=1}^T\frac{D_t}{P_t}\right).
\end{equation*}
Letting $T\to \infty$, we have $\lim_{T\to\infty}q_TP_T>0$ if and only if $\sum_{t=1}^\infty D_t/P_t<\infty$.
\end{proof}

Clearly, $\sum_{t=1}^\infty D_t/P_t<\infty$ only if $D_t/P_t\to 0$. In other words, to attach a rational bubble to a dividend-paying asset, the dividend yield $D_t/P_t$ must converge to zero, or the price-dividend ratio $P_t/D_t$ must diverge to infinity. An analogous result also holds in continuous-time models \citep[Appendix B]{HiranoToda2025EJW}. An important implication of the Bubble Characterization Lemma is that as long as the price-dividend ratio remains bounded, rational bubbles attached to dividend-paying assets can never occur, regardless of the model setting.

\subsection{Bubble characterization with variable shares}

So far, we assumed that the asset has an infinite maturity. In reality, firms sometimes issue or repurchase shares. In this section, we study how the variability in shares affects the characterization of bubbles.

As in \S\ref{subsec:characterization}, consider an infinite-horizon, deterministic economy with time-$t$ Arrow-Debreu price denoted by $q_t>0$. Let there be a firm with stock price and dividend per share at time $t$ denoted by $p_t,d_t$. The absence of arbitrage implies
\begin{equation}
    q_tp_t=q_{t+1}(p_{t+1}+d_{t+1}). \label{eq:noarbitrage_stock}
\end{equation}
Let $S_t>0$ be the number of stock shares outstanding at time $t$. Let $C_t$ be the free cash flow that the firm generates at time $t$. Then accounting implies
\begin{equation}
    C_t=\underbrace{d_tS_{t-1}}_\text{total dividend}+\underbrace{p_t(S_{t-1}-S_t)}_\text{net stock repurchase}. \label{eq:cashflow}
\end{equation}
Note that \eqref{eq:cashflow} can be rewritten as
\begin{equation}
    p_tS_t+C_t=(p_t+d_t)S_{t-1}. \label{eq:cashflow2}
\end{equation}
Define the firm value by the total market capitalization, which is $P_t\coloneqq p_tS_t$. Multiplying $S_t$ to both sides of \eqref{eq:noarbitrage_stock} and using \eqref{eq:cashflow2}, we obtain
\begin{align}
    q_tP_t=q_tp_tS_t&=q_{t+1}(p_{t+1}+d_{t+1})S_t\notag \\
    &=q_{t+1}(p_{t+1}S_{t+1}+C_{t+1})\notag \\
    &=q_{t+1}(P_{t+1}+C_{t+1}). \label{eq:noarbitrage_firm}
\end{align}
By comparing the no-arbitrage conditions \eqref{eq:noarbitrage}, \eqref{eq:noarbitrage_stock}, \eqref{eq:noarbitrage_firm}, and applying Lemma \ref{lem:bubble}, it is clear that
\begin{subequations}
    \begin{align}
        \exists\text{bubble in firm stock}&\iff \lim_{t\to\infty}q_tp_t>0\iff \sum_{t=1}^\infty \frac{d_t}{p_t}<\infty, \label{eq:bubble_stock} \\
        \exists\text{bubble in firm value}&\iff \lim_{t\to\infty}q_tP_t>0\iff \sum_{t=1}^\infty \frac{C_t}{P_t}<\infty. \label{eq:bubble_value}
    \end{align}
\end{subequations}
Noting that $P_t=p_tS_t$, the conditions \eqref{eq:bubble_stock} and \eqref{eq:bubble_value} are not necessarily equivalent. Indeed, we have the following proposition.

\begin{prop}\label{prop:bubble_MM}
The following statements are true.
\begin{enumerate}
    \item\label{item:MM1} If $0<\liminf_{t\to\infty}S_t\le \limsup_{t\to\infty}S_t<\infty$, then there is a bubble in firm stock if and only if there is a bubble in firm value.
    \item\label{item:MM2} If $\liminf_{t\to\infty}S_t=0$, there is no bubble in firm value.
    \item\label{item:MM3} If $\limsup_{t\to\infty}S_t=\infty$, there is no bubble in firm stock.
\end{enumerate}
\end{prop}

The following two examples show that, when condition \ref{item:MM1} does not hold, it is possible that a bubble emerges in either the firm value or the firm stock but not in the other.

\begin{exmp}[Bubbly stock in fundamental firm]
Let $R>1$ be the (constant) gross risk-free rate and let the firm generate the (constant) cash flow $C>0$. The fundamental value of the firm is then constant at
\begin{equation*}
    P_0=\sum_{t=1}^\infty R^{-t}C=\frac{C}{R-1}.
\end{equation*}
Let the firm pay no dividends ($d_t=0$) and repurchase stocks using cash flows so that $S_t=R^{-t}$. Then the accounting identity \eqref{eq:cashflow} implies
\begin{equation*}
    p_t=\frac{C}{S_{t-1}-S_t}=\frac{CR^t}{R-1}=P_0R^t.
\end{equation*}
Obviously, $P_t=P_0=p_tS_t$ is satisfied. Because the stock pays no dividends but has a positive price, by definition there is a bubble. Yet the firm value reflects fundamentals. This example can be thought of as Berkshire Hathaway, which pays no dividends.
\end{exmp}

\begin{exmp}[Fundamental stock in bubbly firm]
Let $R>1$ be the (constant) gross risk-free rate and let the firm generate the (constant) cash flow $C>0$. Define the firm value by
\begin{equation*}
    P_t=\frac{C}{R-1}+bR^t,
\end{equation*}
where $b>0$ reflects a firm value bubble. Let the number of shares be $S_t=R^t$. Then the stock price is
\begin{equation*}
    p_t=\frac{P_t}{S_t}=b+\frac{C}{R-1}R^{-t}.
\end{equation*}
The accounting identity \eqref{eq:cashflow} implies
\begin{equation*}
    d_t=\frac{C-p_t(S_{t-1}-S_t)}{S_{t-1}}=(R-1)b+C(R+1)R^{-t}.
\end{equation*}
Since $d_t/p_t\to R-1>0$ as $t\to\infty$, by \eqref{eq:bubble_stock} the stock price reflects fundamentals.
\end{exmp}

\section{Example economies}\label{sec:example}

To convince skeptics that bubbles are possible in relatively standard macro-finance models, and sometimes inevitable, this section presents several examples with bubbles attached to real assets.

\subsection{OLG model with log utility}\label{subsec:example_log}

The first example, which appears in \citet[\S III.A]{HiranoToda2025JPE}, is a simple variant of the \citet{Samuelson1958} overlapping generations (OLG) model with money, except that the asset pays dividends that are shrinking relative to the endowments in the economy.

The initial old are endowed with a unit supply of an asset with infinite maturity. At time $t$, the young are endowed with $a_t>0$ units of the consumption good, the old none, and the asset pays dividend $D_t>0$. Generation $t$ has utility function
\begin{equation}
    U(y_t,z_{t+1})=(1-\beta)\log y_t+\beta \log z_{t+1}, \label{eq:logutility}
\end{equation}
where $(y_t,z_{t+1})$ denote the consumption when young and old.\footnote{The notation may be nonstandard, but the mnemonic is that $y$ is the first letter of ``young'' and $z$ is the next alphabet.} A competitive equilibrium with sequential trading is defined by a sequence $\set{(P_t,x_t,y_t,z_t)}_{t=0}^\infty$ of asset price $P_t$, asset holdings of young $x_t$, and consumption of young and old $(y_t,z_t)$ such that
\begin{enumerate*}
    \item the young maximize utility subject to the budget constraints $y_t+P_tx_t=a_t$ and $z_{t+1}=(P_{t+1}+D_{t+1})x_t$,
    \item commodity market clears: $y_t+z_t=a_t+D_t$, and
    \item asset market clears: $x_t=1$.
\end{enumerate*}
The following proposition provides a necessary and sufficient condition for bubbles.

\begin{prop}\label{prop:textbook}
There exists a unique equilibrium, and the asset price exhibits a bubble if and only if $\sum_{t=1}^\infty D_t/a_t<\infty$.
\end{prop}

\begin{proof}
Due to log utility, the optimal consumption of the young is $y_t=(1-\beta)a_t$. Asset market clearing and the budget constraint of the young imply $P_t=P_tx_t=a_t-y_t=\beta a_t$. Clearly, the equilibrium is unique. Since the dividend yield is $D_t/P_t=D_t/(\beta a_t)$, the claim follows from Lemma \ref{lem:bubble}.
\end{proof}

A similar example to Proposition \ref{prop:textbook} appears in \citet[\S IV]{AllenBarlevyGale2025} with $\beta=1$ (no young consumption).

\subsection{OLG model with linear utility}\label{subsec:example_linear}

The second example is based on that in \citet[\S7]{Wilson1981}. As far as we are aware, this is the first example of a rational bubble attached to a dividend-paying asset.

This example is similar to \S\ref{subsec:example_log} except that the utility function
\begin{equation*}
    U(y_t,z_{t+1})=y_t+\beta z_{t+1}
\end{equation*}
is linear, endowments are $(a_t,b_t)=(aG^t,bG^t)$ with $a>0$, $b\ge 0$, and dividends are $D_t=DG_d^t$ with $D>0$, $G_d>0$. The following proposition provides a sufficient condition for the uniqueness of equilibrium and the necessity of bubbles. In what follows, longer proofs are deferred to Appendix \ref{sec:proof}.

\begin{prop}\label{prop:wilson}
If $1/\beta<G_d<G$, then the unique equilibrium asset price is $P_t=aG^t$, and there is a bubble.
\end{prop}

\subsection{OLG model with capital and labor}

In an influential paper, \citet{Tirole1985} studied under what conditions bubbles emerge in an overlapping generations (OLG) model with capital accumulation and showed that bubbles can solve the well-known capital over-accumulation problem. The idea is to introduce a dividend-paying asset in the \citet{Diamond1965} model. \citet[Proposition 1]{Tirole1985} connects the population growth rate, dividend growth rate, and the natural interest rate (the steady state interest rate in the \citet{Diamond1965} model without the asset) to the existence of a unique bubbleless equilibrium (part (a)), the existence of a continuum of equilibria including bubbleless and asymptotically bubbly one (part (b)), and the existence of a unique asymptotically bubbly equilibrium (part (c)).

\citet*{BosiHa-HuyLeVanPhamPham2018} extend \citet{Tirole1985}'s overlapping generations production economy with capital and labor to the case with altruism (which is not essential for bubbles) and a dividend-paying asset. Their Proposition 2 derives properties of equilibria with general utility and production functions. By specializing to the Cobb-Douglas utility and production functions, their Example 2 provides bubbly equilibria with a dividend-paying asset. \citet[\S V.A]{HiranoToda2025JPE} study \citet{Tirole1985}'s model with a dividend-paying asset. With log utility and general production function, their Theorem 3 shows that equilibria with $\liminf_{t\to\infty}K_t>0$ are bubbly under some conditions on the dividend growth rate. However, Example 1 of \citet*{BosiHa-HuyLeVanPhamPham2018} shows a case with $\liminf_{t\to\infty}K_t=0$, which is not considered in \citet{Tirole1985}'s original analysis. This suggests that \citet{Tirole1985}'s original proof may be incomplete. In fact, in a recent paper, \citet{PhamTodaComment} construct a counterexample to Proposition 1(c) of \citet{Tirole1985} based on a closed-form solution but restore it under the additional assumption that initial capital is large enough and dividends are small enough. \citet{PhamTodaTirole} discuss other results including the existence of equilibrium and sufficient conditions for the existence or nonexistence of equilibria with or without bubbles.

\citet*{JiangMiaoZhang2022} study a two-sector overlapping generations model in which one sector produces the consumption good from capital and labor, and the other sector produces housing (which serves as a bubble asset) using land and labor. There is spillover between the two sectors through infrastructure investment. Their \S3 analytically solves a two-period OLG model in which housing is a pure bubble asset. In \S4, for a quantitative application, the authors extend the model to a $T$-period OLG model with rents and ``stochastic bubbles''. However, rents are exogenous (writing on p.~1204 ``We do not consider endogenous rents because they would make a bubbly model very hard to analyze''). Moreover, their results are only numerical, and hence they fit  into our ``cautionary warning'' in the introduction. These numerical exercises could be taken  as suggestive of the existence of stochastic bubbles and multiple equilibria (either fundamental or bubbly).

\subsection{Infinite-horizon model}\label{subsec:example_bewley}

In general, it is more difficult to generate asset price bubbles in infinite-horizon models than in OLG models. This is because in a model with infinitely-lived agents and short-sales constraints, if a bubbly equilibrium exists, then there exist no agent who can permanently reduce asset holdings.\footnote{The formal statement appears in \citet[Proposition 3]{Kocherlakota1992}. \citet[Theorem 4.1]{Kamihigashi2018} extends this result to a general setting assuming only the monotonicity of preferences. In OLG models with finite lives, there is no agent who can permanently reduce asset holdings because the old liquidate asset holdings before exiting the economy. The short-sales constraint is implicit in OLG models.} In other words, the short-sales constraint must bind infinitely often, implying that financial constraints are essential for generating asset price bubbles in infinite-horizon economies.

Here, to make the analysis self-contained, we briefly explain \citet{Bewley1980}'s infinite-horizon, two-agent model with alternating endowments. (See also \citet[\S3.2]{HiranoToda2024JME}.) The agents have the utility function
\begin{equation}
    \sum_{t=0}^\infty \beta^t u(c_t), \label{eq:utility}
\end{equation}
where $\beta\in (0,1)$ and $u:[0,\infty)\to [-\infty,\infty)$ is twice differentiable on $(0,\infty)$ with $u'>0$, $u''<0$, $u'(0)=\infty$, and $u'(\infty)=0$. Suppose that there are two agents with endowments alternating as follows:
\begin{align*}
    &\text{Time:} &&(0,1,2,3,\dotsc),\\
    &\text{Agent 1:} && (a,b,a,b,\dotsc),\\
    &\text{Agent 2:} &&(b,a,b,a,\dotsc),
\end{align*}
where $a>b\ge 0$. Suppose there is a unit supply of intrinsically worthless asset (money), which is initially held by agent 2. Suppose the asset cannot be shorted. An equilibrium is defined by sequences of consumption allocations and asset prices such that agents optimize and markets clear. We omit the details as they are standard.

We seek an equilibrium in which the asset trades at a constant price $P>0$. At any date $t$, call the agent with endowment $a$ ($b$) ``rich'' (``poor''). In this economy, because endowments are alternating between high and low values, the rich agent has an incentive to save. Therefore conjecture that the rich agent buys the entire asset from the poor agent, and hence the equilibrium consumption allocation is
\begin{subequations}\label{eq:bewley_allocation}
\begin{align}
    &\text{Agent 1:} && (a-P,b+P,a-P,b+P,\dotsc),\\
    &\text{Agent 2:} &&(b+P,a-P,b+P,a-P,\dotsc).
\end{align}
\end{subequations}
Under this conjecture, because the rich agent holds a long position of the asset, the Euler equation must hold:
\begin{equation}
    u'(a-P)=\beta u'(b+P). \label{eq:euler_rich}
\end{equation}
Because the short-sales constraint binds for the poor agent, the Euler inequality becomes
\begin{equation}
    u'(b+P)\ge \beta u'(a-P). \label{eq:euler_poor}
\end{equation}
The following proposition shows that, when $a$ is sufficiently high, there exists a unique $P>0$ satisfying these conditions.

\begin{prop}\label{prop:bewley}
If $u'(a)<\beta u'(b)$, there exists a unique $P\in (0,a)$ satisfying \eqref{eq:euler_rich} and \eqref{eq:euler_poor}. The allocation \eqref{eq:bewley_allocation} together with asset price $P>0$ constitute an equilibrium.
\end{prop}

There are many results based on \citet{Bewley1980}'s model in the literature, including \citet{ScheinkmanWeiss1986}, \citet{Woodford1990}, \citet[Example 1]{Kocherlakota1992}, \citet[Example 7.1]{HuangWerner2000}, and \citet[Example 1]{Werner2014}, which are all pure bubble models without dividends. \citet{LeVanPham2016} consider a model with both physical capital and a dividend-paying asset, and they provide an example of bubbles attached to the dividend-paying asset in \S6.1.2 in a fairly limited setting (the production function is linear with respect to capital and labor). See also \citet*[\S4.2]{BosiLeVanPham2017MSS} and \citet*[\S5]{BosiLeVanPham2018} for related results. \citet*[\S4.1]{BosiLeVanPham2022} extend \citet{Bewley1980}'s model with general endowments. Their Proposition 7 and the subsequent discussion construct bubbly equilibria with a dividend-paying asset. 

Because the example of \citet*{BosiLeVanPham2022} is rather involved, here we present a simple example based on an earlier version of \citet*{HiranoJinnaiTodaLeverage},\footnote{Specifically, See \S2.2.2 of \url{https://arxiv.org/abs/2211.13100v4}.} which is an extension of Example 1 of \citet{Kocherlakota1992}. Let there be two agents with utility function \eqref{eq:utility}, where the period utility takes the constant relative risk aversion (CRRA) form
\begin{equation*}
    u(c)=\begin{cases*}
        \frac{c^{1-\gamma}}{1-\gamma} & if $0<\gamma\neq 1$,\\
        \log c & if $\gamma=1$.
    \end{cases*}
\end{equation*}
There is a unit supply of a long-lived asset that pays a constant dividend $D>0$ in every period. The aggregate endowment at time $t$ (including dividend) is $(a+b)G^t$, where $G>1$ and $a>b>0$. The asset is initially owned by agent 2. Suppose the asset cannot be shorted.

We specify individual endowments such that agent 1 is rich (poor) in even (odd) periods, and vice versa for agent 2. Conjecture that in equilibrium, individual consumption is
\begin{equation*}
(c_{1t},c_{2t})=\begin{cases*}
    ((a-p)G^t,(b+p)G^t) & if $t$: even,\\
    ((b+p)G^t,(a-p)G^t) & if $t$: odd
    \end{cases*}
\end{equation*}
for some $0\le p<a$. Conjecture that the asset price at time $t$ is
\begin{equation}
    P_t=\frac{D}{G-1}+pG^t, \label{eq:Pt_bewley}
\end{equation}
where we conjecture that the gross risk-free rate is $R=G$, $\frac{D}{G-1}=\sum_{t=1}^\infty R^{-t}D$ is the fundamental value of the asset, and $pG^t$ is the bubble component. Conjecture that every period, the poor (rich) agent sells (buys) the entire asset to smooth consumption. Letting $e_t^r$ ($e_t^p$) be the time $t$ endowment of the rich (poor) agent, the budget constraints imply
\begin{align*}
    &\text{Rich:} && (a-p)G^t+P_t\cdot 1 = (P_t+D)\cdot 0 + e_t^r \iff e_t^r=aG^t+\frac{1}{G-1}D,\\
    &\text{Poor:} && (b+p)G^t+P_t\cdot 0 = (P_t+D)\cdot 1 + e_t^p \iff e_t^p=bG^t-\frac{G}{G-1}D.
\end{align*}
Let $D>0$ be small enough such that $e_0^p=b-\frac{G}{G-1}D>0$, which implies $e_t^p>0$ for all $t$ because $G>1$. Since the rich agent is unconstrained, the Euler equation must hold with equality. For the poor agent, the Euler equation may be an inequality. Since by assumption we have $R=G$, the Euler equations become
\begin{align*}
    &\text{Rich:} && \beta G \left(\frac{b+p}{a-p}G\right)^{-\gamma}=1,\\
    &\text{Poor:} && \beta G \left(\frac{a-p}{b+p}G\right)^{-\gamma}\le 1.
\end{align*}
Solving the Euler equation of the rich, we obtain
\begin{equation}
    p=\frac{a(\beta G^{1-\gamma})^{1/\gamma}-b}{1+(\beta G^{1-\gamma})^{1/\gamma}}. \label{eq:p}
\end{equation}
For $p>0$, it is necessary and sufficient that $\beta G^{1-\gamma}>(b/a)^\gamma$. For the Euler inequality for the poor agent to hold, it is necessary and sufficient that
\begin{equation}
    1\ge \left(\frac{b+p}{a-p}\right)^\gamma\beta G^{1-\gamma}=(\beta G^{1-\gamma})^2\iff \beta G^{1-\gamma}\le 1. \label{eq:Euler_poor}
\end{equation}
To show that we have an equilibrium, it suffices to show the transversality condition for optimality $\lim_{t\to\infty}\beta^t u'(c_t)P_t=0$  \citep[p.~237, Example 15.3]{TodaEME}, where $c_t$ is the consumption of any agent. Since $c_t\sim G^t$ and $P_t\sim G^t$ as $t\to\infty$, we obtain $\beta^t u'(c_t)P_t\sim (\beta G^{1-\gamma})^t\to 0$ if and only if $\beta G^{1-\gamma}<1$, in which case the Euler inequality for the poor \eqref{eq:Euler_poor} holds. Therefore we obtain the following proposition.

\begin{prop}\label{prop:exmp_infhor}
Let $\beta\in (0,1)$ and $\gamma>0$ be given. Take any $G>1$ such that $\beta G^{1-\gamma}<1$. Take any $a,b,D>0$ such that
\begin{equation*}
    \frac{G}{G-1}D<b<(\beta G^{1-\gamma})^{1/\gamma}a
\end{equation*}
holds and define $p>0$ by \eqref{eq:p}. Then the consumption allocation $(c_t^r,c_t^p)=((a-p)G^t,(b+p)G^t)$ and asset price $P_t=\frac{D}{G-1}+pG^t$ constitute a bubbly equilibrium.
\end{prop}

\subsection{Generality and economic relevance}\label{subsec:generality}

So far, we have seen several example economies with bubbles attached to real assets. However, these examples are shown in fairly limited settings. Therefore, it is not obvious to what extent there is generality and how economically relevant the results are, nor is it obvious what new insights and asset pricing implications can be drawn when we consider more general macro-finance models. From an economic perspective, it would be fair to say that these questions are far more important than just proving the existence of a bubble in one setting or another.\footnote{Another reaction from the general audience when one of the authors (Hirano) presented earlier papers on pure bubbles was that researchers working on bubbles are preoccupied with showing that bubbles can or cannot occur in certain limited settings just from theoretical curiosity, rather than think about the generality of results and the economic implications. As we discuss below, this common view of the literature is totally wrong.}

Our series of papers \citep*{HiranoToda2025JPE,HiranoToda2025PNAS,HiranoTodaHousingbubble,HiranoToda2024JME,HiranoJinnaiTodaLeverage} address these questions head-on and argue that the theory of asset price bubbles attached to real assets is closely related to the root of economic development. \citet[\S III.B]{HiranoToda2025JPE} consider a two-sector production economy with land and uneven productivity growth and show that land bubbles necessarily emerge if the productivity growth is faster in the non-land sector. \citet{HiranoToda2025PNAS} significantly extend this result under aggregate uncertainty and show that land prices exhibit recurrent stochastic fluctuations along economic development, with expansions and contractions in the size of land bubbles. \citet[\S III.C]{HiranoToda2025JPE} consider a production economy with capital and labor and show the necessity of stock price (capital) bubble under some condition on the elasticity of substitution and productivity growth. \citet*{HiranoJinnaiTodaLeverage} consider an infinite-horizon macro-finance model and show that once financial leverage or overall productivity of the economy gets sufficiently high, the dynamic path dramatically changes and deviates from the balanced growth path, necessarily leading to land price bubbles. \citet[\S6]{HiranoToda2024JME} study a special case with a closed-form solution with linear production.

As can be seen from these results, once we consider asset price bubbles attached to real assets in more general macro-finance models, we can derive new insights. One is the concept of the \emph{necessity} of bubbles, and the other is the importance of \emph{unbalanced growth}. The concept of the necessity of bubbles is fundamentally different from the concept of the possibility of bubbles as in pure bubble models, \ie, bubbles \emph{can} arise under some conditions. Bubble necessity means that there exist neither fundamental equilibria nor bubbly equilibria that become asymptotically bubbleless, and all equilibria \emph{must} be asymptotically bubbly. \citet{HiranoToda2025JPE} prove that the necessity of bubbles can be widely obtained in workhorse macroeconomic models, including Bewley models with idiosyncratic investment shocks (their \S V.B) and preference shocks (their \S V.C). Unbalanced growth means that different factors of production or different sectors have different productivity growth rates. Hence, unbalanced growth entails a world of nonstationarity. It is well known that the conventional macroeconomic theory with balanced growth requires knife-edge restrictions implied by the Uzawa balanced growth theorem \citep{Uzawa1961,Schlicht2006,GrossmanHelpmanOberfieldSampson2017}. Once we remove these restrictions and consider the global parameter space from the outset, rather than focusing on the knife-edge case, the implications for asset pricing dramatically change.

In the rest of this chapter, we address the concept of the necessity of bubbles established in \citet{HiranoToda2025JPE}. Although the results are based on our earlier work, we explain the concept in slightly different models.

\section{Necessity of bubbles}\label{sec:necessity}

We consider the standard two-period overlapping generations model. Let $U(y,z)$ denote the utility function of a typical agent, where $(y,z)$ denote the consumption when young and old. We assume that $U$ is quasi-concave, differentiable with positive partial derivatives, and satisfies the Inada condition. The endowments of the young are old at time $t$ are denoted by $(a_t,b_t)$, where $a_t>0$ and $b_t\ge 0$. There is a dividend-paying asset with infinite maturity in unit supply, which is initially owned by the old. Let $D_t\ge 0$ be the dividend at time $t$, with $D_t>0$ infinitely often to guarantee that the asset price is always strictly positive.

Letting $P_t>0$ be the asset price (in units of the date-$t$ good) and $x_t$ the number of asset shares demanded by the young, the budget constraints are
\begin{subequations}\label{eq:budget}
\begin{align}
    &\text{Young:} & y_t+P_tx_t&=a_t, \label{eq:budget_young}\\
    &\text{Old:} & z_{t+1}&=b_{t+1}+(P_{t+1}+D_{t+1})x_t. \label{eq:budget_old}
\end{align}
\end{subequations}
Solving for $(y_t,z_{t+1})$, the utility maximization problem of generation $t$ is
\begin{equation}
    \max_x U(a_t-P_tx,b_{t+1}+(P_{t+1}+D_{t+1})x), \label{eq:UMP}
\end{equation}
where $x\le a_t/P_t$ to prevent negative consumption.

A rational expectations equilibrium is defined by a sequence of prices and allocations $\set{(P_t,x_t,y_t,z_t)}_{t=0}^\infty$ such that all agents optimize and the commodity and asset markets clear. Regarding the asset market, because the old exit the economy and hence liquidate their asset holdings, the young are the natural buyer. Therefore the asset market clearing condition is $x_t=1$.

Let
\begin{equation}
    (y_t,z_{t+1})=(a_t-P_t,b_{t+1}+P_{t+1}+D_{t+1}) \label{eq:yz}
\end{equation}
be the consumption of generation $t$ obtained by the budget constraint \eqref{eq:budget} and imposing the asset market clearing condition $x_t=1$. The first-order condition of the utility maximization problem \eqref{eq:UMP} (Euler equation) evaluated at the equilibrium allocation $x_t=1$ is
\begin{equation}
    U_y(y_t,z_{t+1})P_t=U_z(y_t,z_{t+1})(P_{t+1}+D_{t+1}), \label{eq:foc}
\end{equation}
where $(y_t,z_{t+1})$ is as in \eqref{eq:yz}. A standard truncation argument \citep{BalaskoShell1980} implies the existence of equilibrium. Furthermore, because $D_t>0$ infinitely often, we necessarily have $P_t>0$: see the proof of Theorem 1 of \citet{HiranoToda2025JPE}. Another useful property is that given $P_{t+1}>0$, there exists a unique $P_t>0$ satisfying \eqref{eq:foc}. We state this result as a lemma.

\begin{lem}\label{lem:backward}
For any $P_{t+1}>0$, there exists a unique $P_t\in (0,a_t)$ satisfying \eqref{eq:foc}.
\end{lem}

Lemma \ref{lem:backward} allows us to extend an equilibrium backwards in time uniquely, thereby allowing us to focus on the long run behavior of the equilibrium.

In any equilibrium, we may define the gross risk-free rate between time $t$ and $t+1$ by
\begin{equation}
    R_t\coloneqq \frac{P_{t+1}+D_{t+1}}{P_t}=\frac{U_y}{U_z}(y_t,z_{t+1}). \label{eq:Rt}
\end{equation}
Define the Arrow-Debreu price (data-0 price of the consumption good delivered at time $t$) by $q_0=1$ and $q_t=1/\prod_{s=0}^{t-1}R_s$ for all $t>0$. By the discussion in \S\ref{sec:bubble}, the fundamental value of the asset is then $V_0=\sum_{t=1}^\infty q_tD_t$.

We now state a result implying the necessity of bubbles. Suppose for simplicity that the endowments are stationary, so $(a_t,b_t)=(a,b)$ for all $t$. Define the long-run dividend growth rate by
\begin{equation}
    G_d\coloneqq \limsup_{t\to\infty} D_t^{1/t} \label{eq:Gd}
\end{equation}
and the quantity
\begin{equation}
    R\coloneqq \frac{U_y}{U_z}(a,b). \label{eq:R_aut}
\end{equation}
Using \eqref{eq:Rt}, note that $R$ in \eqref{eq:R_aut} is the gross risk-free rate that would prevail in a counterfactual economy without the asset.

\begin{thm}\label{thm:necessity_OLG}
If $R<G_d<1$, then all equilibria are bubbly with $\liminf_{t\to\infty}P_t>0$.
\end{thm}

Theorem \ref{thm:necessity_OLG} is a special case of \citet[Theorem 2]{HiranoToda2025JPE} with $(a_t,b_t)=(a,b)$ and $G=1$, so we omit the proof (which is technical). Here we explain the intuition. Because dividends grow at rate $G_d$ in \eqref{eq:Gd}, if a fundamental equilibrium exists, the asset price $P_t$ grows at the same rate of $G_d<1$ and hence converges to 0. By \eqref{eq:yz}, the equilibrium allocation $(y_t,z_{t+1})$ converges to $(a,b)$. Then the interest rate $R_t$ in \eqref{eq:Rt} converges to the counterfactual autarky interest rate $R$ in \eqref{eq:R_aut}. But since by assumption $R<G_d$, the fundamental value of the asset (the present value of dividends) becomes infinite, which is impossible. Therefore fundamental equilibria do not exist.

Of course, this argument is heuristic because the part ``the asset price $P_t$ grows at the same rate of $G_d<1$'' is not obvious. The actual proof of Theorem 2 of \citet{HiranoToda2025JPE} avoids this issue by showing that all equilibria satisfy the properties stated in Theorem \ref{thm:necessity_OLG} without relying on the convergence of $P_t$.

\section{Long-run behavior of asset prices}\label{sec:longrun}

As noted in \S\ref{subsec:generality}, bubbles attached to real assets entail a nonstationary world with unbalanced growth. Dealing with this world requires analytical tools. Since the asset price is a forward-looking variable and economic agents are rational, as long as bubbles are expected to arise in the future, by a backward induction argument, bubbles will arise at present. Thus whether bubbles emerge in the future depends on the long-run behavior of the model. In this section, we explain how to study such models quantitatively by applying the local stable manifold theorem.

\subsection{Model}

The model is a special case of the OLG model in \S\ref{sec:necessity}. In addition, we assume that the utility function $U$ is increasing, quasi-concave, and homothetic. Without loss of generality, assume $U$ is homogeneous of degree 1. Then Theorem 11.14 of \citet[p.~158]{TodaEME} implies that $U$ is actually concave. Because we shall use calculus, we impose the following regularity conditions.

\begin{asmp}\label{asmp:U}
    The utility function $U:\R_{++}^2\to (0,\infty)$ is homogeneous of degree 1, twice continuously differentiable, and satisfies $U_y>0$, $U_z>0$, $U_{yy}<0$, $U_{zz}<0$, $U_y(0,z)=\infty$, $U_z(y,0)=\infty$.
\end{asmp}

Furthermore, we specialize the endowments and dividends as follows.

\begin{asmp}\label{asmp:G}
    The date-$t$ endowments of the young and old are denoted by $(a_t,b_t)=(aG^t,bG^t)$, where $G>0$ is the economic growth rate and $a>0,b\ge 0$. The date-$t$ dividend is denoted by $D_t=DG_d^t$, where $G_d\in (0,G)$ is the dividend growth rate and $D>0$.
\end{asmp}

The condition $D>0$ implies that the asset pays dividends, unlike pure bubble models in the literature. The condition $G_d<G$ is important for generating asset price bubbles. (In \S\ref{sec:necessity}, we had $G=1$.) Below, we focus on equilibria in which the asset price grows at an asymptotically constant rate.

\subsection{Fundamental equilibria}

Suppose first that the asset price reflects fundamentals, so $P_t=V_t$.

\paragraph{Derivation of autonomous system}

Since by Assumption \ref{asmp:G} dividends grow at rate $G_d$, we may conjecture that so does $P_t=V_t$. This motivates us to define the detrended asset price $p_t\coloneqq G_d^{-t}P_t$. Dividing the first-order condition \eqref{eq:foc} by $G_d^t$, we obtain
\begin{equation}
    U_yp_t=G_dU_z(p_{t+1}+D), \label{eq:foc_f1}
\end{equation}
where $U_y,U_z$ are evaluated at
\begin{equation}
    (y,z)=(aG^t-p_tG_d^t,bG^{t+1}+(p_{t+1}+D)G_d^{t+1}). \label{eq:yz_f1}
\end{equation}
By Assumption \ref{asmp:U}, $U$ is homogeneous of degree 1 and hence $U_y,U_z$ are homogeneous of degree 0. Therefore by dividing \eqref{eq:yz_f1} by $G^t$, \eqref{eq:foc_f1} remains valid by evaluating at
\begin{equation}
    (y,z)=(a-p_t(G_d/G)^t,Gb+G_d(p_{t+1}+D)(G_d/G)^t). \label{eq:yz_f2}
\end{equation}
The nonlinear difference equation \eqref{eq:foc_f1} explicitly depends on time because $(G_d/G)^t$ enters in the arguments. Thus, the system is non-autonomous, which is inconvenient for analysis. To remove the explicit dependence on time, we introduce the auxiliary variable $\xi_t=(\xi_{1t},\xi_{2t})\in \R_{++}^2$ defined by $\xi_{1t}=p_t=P_t/G_d^t$ and $\xi_{2t}=(G_d/G)^t$. Then we can write the system as $\Phi(\xi_t,\xi_{t+1})=0$, where $\Phi:\R^4\to \R^2$ is given by
\begin{subequations}\label{eq:Phi_f}
    \begin{align}
        \Phi_1(\xi,\eta)&=G_d(\eta_1+D)U_z-\xi_1U_y, \label{eq:Phi1_f}\\
        \Phi_2(\xi,\eta)&=\eta_2-(G_d/G)\xi_2, \label{eq:Phi2_f}
    \end{align}
\end{subequations}
where $(\xi,\eta)=(\xi_1,\xi_2,\eta_1,\eta_2)$ and $\Phi=(\Phi_1,\Phi_2)$. In \eqref{eq:Phi_f}, using \eqref{eq:yz_f2} and the definition of $\xi_t$, the partial derivatives $U_y,U_z$ are evaluated at
\begin{equation*}
    (y,z)=(a-\xi_1\xi_2,Gb+G_d(\eta_1+D)\xi_2).
\end{equation*}

\paragraph{Steady state}

Let $\xi^*$ be a steady state of the system $\Phi(\xi_t,\xi_{t+1})=0$ defined by $\Phi(\xi^*,\xi^*)=0$. Noting that $G_d\in (0,G)$, \eqref{eq:Phi2_f} implies $\xi_2^*=0$. Then \eqref{eq:Phi1_f} implies
\begin{equation}
    G_d(\xi_1^*+D)U_z-\xi_1^*U_y=0\iff \xi_1^*=\frac{G_dDU_z}{U_y-G_dU_z}, \label{eq:xi_1f}
\end{equation}
where $U_y,U_z$ are evaluated at $(y,z)=(a,Gb)$. Because $\xi_{1t}=p_t$ is a normalized price, it must be positive. Therefore a necessary and sufficient condition for the existence of a steady state is
\begin{equation}
    U_y-G_dU_z>0\iff G_d<\frac{U_y}{U_z}(a,Gb). \label{eq:cond_f}
\end{equation}
The economic intuition for the existence condition \eqref{eq:cond_f} is the following. If the asset price reflects fundamentals, because dividends grow at rate $G_d$, so does the asset price. Because endowments grow at a higher rate $G>G_d$, the asset price becomes negligible in the long run, and the consumption allocation approaches autarky. Using \eqref{eq:Rt}, the long run interest rate converges to the right-hand side of \eqref{eq:cond_f}. In equilibrium, this interest rate must exceed the dividend growth rate, for otherwise the fundamental value is infinite, which is impossible.

\paragraph{Asymptotic behavior}

To study the asymptotic behavior of the solution to the nonlinear implicit difference equation $\Phi(\xi_t,\xi_{t+1})=0$, we apply the implicit function theorem and the local stable manifold theorem. We first solve the nonlinear equation $\Phi(\xi,\eta)=0$ as $\eta=\phi(\xi)$ near the steady state $(\xi,\eta)=(\xi^*,\xi^*)$ applying the implicit function theorem. Differentiating \eqref{eq:Phi_f} with respect to $\xi$, we obtain
\begin{align*}
    \frac{\partial \Phi_1}{\partial \xi_1}&=G_d(\eta_1+D)(-\xi_2 U_{yz})-U_y+\xi_1\xi_2U_{yy},\\
    \frac{\partial \Phi_1}{\partial \xi_2}&=G_d(\eta_1+D)(-\xi_1 U_{yz}+G_d(\eta_1+D)U_{zz})-\xi_1(-\xi_1U_{yy}+G_d(\eta_1+D)U_{yz}),\\
    \frac{\partial \Phi_2}{\partial \xi_1}&=0,\\
    \frac{\partial \Phi_2}{\partial \xi_2}&=-G_d/G.
\end{align*}
Evaluating these partial derivatives at $\xi^*=(\xi_1^*,0)$, we obtain the Jacobian
\begin{equation*}
    D_\xi \Phi(\xi^*,\xi^*)=
    \begin{bmatrix}
        -U_y & (\xi_1^*)^2U_{yy}-2\xi_1^*G_d(\xi_1^*+D)U_{yz}+[G_d(\xi_1^*+D)]^2U_{zz} \\ 0 & -G_d/G
    \end{bmatrix}.
\end{equation*}
Similarly, differentiating \eqref{eq:Phi_f} with respect to $\eta$, we obtain
\begin{align*}
    \frac{\partial \Phi_1}{\partial \eta_1}&=G_d(U_z+(\eta_1+D)G_d\xi_2U_{zz})-\xi_1G_d\xi_2U_{yz},\\
    \frac{\partial \Phi_1}{\partial \eta_2}&=0,\\
    \frac{\partial \Phi_2}{\partial \eta_1}&=0,\\
    \frac{\partial \Phi_2}{\partial \eta_2}&=1.
\end{align*}
Evaluating these partial derivatives at $\xi^*=(\xi_1^*,0)$, we obtain the Jacobian
\begin{equation*}
    D_\eta\Phi(\xi^*,\xi^*)=\begin{bmatrix}
        G_dU_z & 0\\
        0 & 1
    \end{bmatrix}.
\end{equation*}
Since $D_\eta\Phi$ is nonsingular, we can apply the implicit function theorem, and we obtain the Jacobian of $\phi$
\begin{equation*}
    D\phi(\xi^*)=-[D_\eta\Phi(\xi^*,\xi^*)]^{-1}D_\xi\Phi(\xi^*,\xi^*)=\begin{bmatrix}
        U_y/(G_dU_z) & * \\
        0 & G_d/G
    \end{bmatrix},
\end{equation*}
where the term in $*$ is unimportant. Condition \eqref{eq:cond_f} implies that the first eigenvalue of $D\phi$ is $\lambda_1\coloneqq U_y/(G_dU_z)>1$. Assumption \ref{asmp:G} implies that the second eigenvalue of $D\phi$ is $\lambda_2\coloneqq G_d/G\in (0,1)$. Therefore the steady state $\xi^*$ is a saddle point and we obtain the following proposition.

\begin{prop}\label{prop:f}
The following statements are true.
\begin{enumerate}
    \item\label{item:f1} There exists a unique $w=w_f^*$ satisfying $(U_y/U_z)(1,Gw)=G_d$.
    \item\label{item:f2} There exists a steady state $\xi^*$ of $\Phi$ in \eqref{eq:Phi_f} if and only if $b/a>w_f^*$. Under this condition, there exists a unique path $\set{\xi_t^*}_{t=0}^\infty$ converging to $\xi^*$.
    \item\label{item:f3} The corresponding equilibrium asset price has order of magnitude
    \begin{equation*}
        P_t=\xi_{1t}^*G_d^t\sim \frac{G_dU_z}{U_y-G_dU_z}DG_d^t,
    \end{equation*}
    and there is no asset price bubble.
\end{enumerate}
\end{prop}

\subsection{Bubbly equilibria}

We next consider bubbly equilibria, so $P_t>V_t$.

\paragraph{Derivation of autonomous system}

In bubbly equilibria, the bubble size need not grow at the same rate as dividends. Therefore, we define the detrended asset price $p_t\coloneqq G^{-t}P_t$. Dividing the first-order condition \eqref{eq:foc} by $G^t$, we obtain
\begin{equation}
    U_yp_t=GU_z(p_{t+1}+D(G_d/G)^{t+1}), \label{eq:foc_b1}
\end{equation}
where $U_y,U_z$ are evaluated at
\begin{equation}
    (y,z)=((a-p_t)G^t,(b+p_{t+1})G^{t+1}+DG_d^{t+1}). \label{eq:yz_b1}
\end{equation}
Dividing \eqref{eq:yz_b1} by $G^t$ and using the homogeneity of $U$, \eqref{eq:foc_b1} remains valid by evaluating at
\begin{equation}
    (y,z)=(a-p_t,G(b+p_{t+1}+D(G_d/G)^{t+1})). \label{eq:yz_b2}
\end{equation}
To derive the autonomous system, define the auxiliary variable $\xi_t=(\xi_{1t},\xi_{2t})\in \R^2_{++}$ by $\xi_{1t}=p_t=P_t/G^t$ and $\xi_{2t}=(G_d/G)^t$. Then we can write the system as $\Phi(\xi_t,\xi_{t+1})=0$, where $\Phi:\R^4\to \R^2$ is given by
\begin{subequations}\label{eq:Phi_b}
    \begin{align}
        \Phi_1(\xi,\eta)&=G(\eta_1+D\eta_2)U_z-\xi_1U_y, \label{eq:Phi1_b}\\
        \Phi_2(\xi,\eta)&=\eta_2-(G_d/G)\xi_2, \label{eq:Phi2_b}
    \end{align}
\end{subequations}
where the partial derivatives $U_y,U_z$ are evaluated at
\begin{equation*}
    (y,z)=(a-\xi_1,G(b+p_1+D\eta_2)).
\end{equation*}

\paragraph{Steady state}

Let $\xi^*$ be a steady state. As in the fundamental case, we have $\xi_2^*=0$. Then \eqref{eq:Phi1_b} implies
\begin{equation}
    G\xi_1^*U_z-\xi_1^*U_y\iff \xi_1^*=0~\text{or}~\frac{U_y}{U_z}(a-\xi_1^*,G(b+\xi_1^*))=G. \label{eq:xi_1b}
\end{equation}
The economic intuition for the second case in the steady state condition \eqref{eq:xi_1b} is the following. If the asset price exhibits a bubble and its size is non-negligible relative to the economy, it must asymptotically grow at the same rate as the economy, $G$. Then the gross risk-free rate \eqref{eq:Rt} converges to $G$, which is equivalent to \eqref{eq:xi_1b}. Below, we refer to the case $\xi_1^*=0$ ($\xi_1^*>0$) as the fundamental (bubbly) steady state.

\paragraph{Asymptotic behavior}

Again, we apply the implicit function theorem and the local stable manifold theorem to study the asymptotic behavior. Differentiating \eqref{eq:Phi_b} with respect to $\xi$, we obtain
\begin{align*}
    \frac{\partial \Phi_1}{\partial \xi_1}&=-G(\eta_1+D\eta_2)U_{yz}-U_y+\xi_1U_{yy},\\
    \frac{\partial \Phi_1}{\partial \xi_2}&=0,\\
    \frac{\partial \Phi_2}{\partial \xi_1}&=0,\\
    \frac{\partial \Phi_2}{\partial \xi_2}&=-G_d/G.
\end{align*}
Evaluating these partial derivatives at $\xi^*=(\xi_1^*,0)$, we obtain the Jacobian
\begin{equation*}
    D_\xi \Phi(\xi^*,\xi^*)=
    \begin{bmatrix}
        -G\xi_1^*U_{yz}-U_y+\xi_1^*U_{yy} & 0 \\
        0 & -G_d/G
    \end{bmatrix}.
\end{equation*}
Similarly, differentiating \eqref{eq:Phi_b} with respect to $\eta$, we obtain
\begin{align*}
    \frac{\partial \Phi_1}{\partial \eta_1}&=G(U_z+G(\eta_1+D\eta_2) U_{zz})-G\xi_1U_{yz},\\
    \frac{\partial \Phi_1}{\partial \eta_2}&=G(DU_z+GD(\eta_1+D\eta_2) U_{zz})-GD\xi_1U_{yz},\\
    \frac{\partial \Phi_2}{\partial \eta_1}&=0,\\
    \frac{\partial \Phi_2}{\partial \eta_2}&=1.
\end{align*}
Evaluating these partial derivatives at $\xi^*=(\xi_1^*,0)$, we obtain the Jacobian
\begin{equation*}
    D_\eta\Phi(\xi^*,\xi^*)=\begin{bmatrix}
        G(U_z+G\xi_1^* U_{zz}-\xi_1^* U_{yz}) & GD(U_z+G\xi_1^* U_{zz}-\xi_1^* U_{yz})\\
        0 & 1
    \end{bmatrix}.
\end{equation*}
To simplify notation, define
\begin{align*}
    d&\coloneqq G(U_z+G\xi_1^*U_{zz}-\xi_1^*U_{yz}),\\
    n&\coloneqq G\xi_1^*U_{yz}+U_y-\xi_1^*U_{yy},
\end{align*}
where the symbols $d,n$ correspond to the denominator and numerator, as we shall see momentarily. Then $D_\eta\Phi$ is nonsingular if and only if $d\neq 0$, and under this condition, we obtain the Jacobian of $\phi$
\begin{equation*}
    D\phi(\xi^*)=-[D_\eta\Phi(\xi^*,\xi^*)]^{-1}D_\xi\Phi(\xi^*,\xi^*)=\begin{bmatrix}
        n/d & -DG_d/G \\
        0 & G_d/G
    \end{bmatrix}.
\end{equation*}
Therefore the eigenvalues of $D\phi(\xi^*)$ are $\lambda_1\coloneqq n/d$ and $\lambda_2=G_d/G\in (0,1)$. The case in which the argument breaks down are when either $d=0$ (the implicit function theorem is inapplicable) or $n/d=\pm 1$ (the local stable manifold theorem is inapplicable). Therefore we obtain the following proposition regarding equilibrium paths converging to the bubbly steady state.

\begin{prop}\label{prop:b}
The following statements are true.
\begin{enumerate}
    \item\label{item:b1} There exists a unique $w=w_b^*>w_f^*$ satisfying $(U_y/U_z)(1,Gw)=G$.
    \item\label{item:b2} There exists a bubbly steady state $\xi^*>0$ of $\Phi$ in \eqref{eq:Phi_b} if and only if $b/a<w_b^*$. Under this condition, there exists a path $\set{\xi_t^*}_{t=0}^\infty$ converging to $\xi^*$ if $d\neq 0,-n$. The path is unique if $d>0$.
    \item\label{item:b3} The corresponding equilibrium asset price has order of magnitude
    \begin{equation*}
        P_t=\xi_{1t}^*G^t\sim \frac{w_b^*a-b}{1+w_b^*}G^t,
    \end{equation*}
    and there is an asset price bubble.
\end{enumerate}
\end{prop}

To complete the analysis, it remains to consider the fundamental steady state $\xi^*=0$. In this case, the eigenvalues of $D\phi(\xi^*)$ are $\lambda_1=n/d=U_y/(GU_z)>0$ and $\lambda_2=G_d/G\in (0,1)$. If $w=b/a<w_b^*$, the definition of $w_b^*$ in Proposition \ref{prop:b} implies that $(U_y/U_z)(a,b)<G$ and hence $\lambda_1<1$, so the fundamental steady state is stable. We thus obtain the following theorem.

\begin{thm}\label{thm:lr}
Let $w=b/a$ be the old-to-young income ratio and define $w_f^*<w_b^*$ as in Propositions \ref{prop:f}, \ref{prop:b}. Then the following statements are true.
\begin{enumerate}
    \item\label{item:lr1} If $w>w_f^*$, there exists a unique equilibrium such that $G_d^{-t}P_t$ converges to a positive number, which is fundamental.
    \item\label{item:lr2} If $w<w_b^*$, there exists an equilibrium such that $G^{-t}P_t$ converges to a positive number, which is bubbly. If in addition $w<w_f^*$, there exist no fundamental equilibria.
    \item\label{item:lr3} If $w_f^*<w<w_b^*$, there exist a continuum of equilibria such that $G^{-t}P_t$ converges to zero, which are all bubbly except the unique equilibrium in \ref{item:lr1}.
\end{enumerate}
\end{thm}

Consider a situation where the incomes of the young rise relative to the incomes of the old. Theorem \ref{thm:lr} implies that asset pricing implications markedly change with economic development. In other words, when the incomes of the young are relatively low, the asset price reflects the fundamental value. When the incomes of the young rise and exceed a critical threshold, the economy enters a new phase in which both bubbly and fundamental equilibria can coexist. Once the incomes of the young reach a still higher critical threshold, the situation changes dramatically. That is, the only possible equilibrium is one that features asset price bubbles. Moreover, the existence of a continuum of equilibria in the intermediate region implies that asset price volatility would be highest with a medium level of economic development.

\section{Necessity of stock and land bubbles}\label{sec:land}

In this section, we consider a model with two assets, stocks and land, and show the necessity of bubbles in aggregate stock and land prices.

\subsection{Model}

The model is essentially a combination of \S III.B, III.C of \citet{HiranoToda2025JPE}. Consider a deterministic two-period OLG economy with a homogeneous good and log utility \eqref{eq:logutility}. There are two sectors, a capital-intensive sector (\eg, manufacturing) and a land-intensive sector (\eg, agriculture). In the capital-intensive sector, a representative firm produces the output using the neoclassical production function $F(K,L)$, where $K,L>0$ denote the capital and labor inputs. For simplicity, we exogenously specify the capital and labor supply at time $t$ as $K_t,L_t>0$.\footnote{We may also consider endogenous capital accumulation and labor supply, but the asset pricing implications are the same. \citet*{HiranoKishiTodaBursting} consider a model of stochastic stock bubbles with endogenous technological innovation.} A stock is a claim to capital rents; let $N>0$ denote the number of shares outstanding and $Q_t>0$ be the stock price at time $t$. In the land-intensive sector, a unit of land produces $D_t>0$ units of output; let $X>0$ denote the aggregate land supply and $P_t>0$ be the land price at time $t$.

The firm takes the capital rental rate $r_t>0$ and wage rate $w_t>0$ as given and maximizes the profit
\begin{equation*}
    F(K,L)-r_tK-w_tL,
\end{equation*}
which implies the first-order conditions $r_t=F_K(K_t,L_t)$ and $w_t=F_L(K_t,L_t)$. The capital rent is paid out as stock dividend, which equals $r_tK_t/N$ per share. The land dividend equals $D_t$ per unit. Let $R_t$ be the gross risk-free rate. Because the economy is deterministic, both the stock and land must yield the same return and the no-arbitrage condition
\begin{equation}
    R_t\coloneqq \frac{Q_{t+1}+r_{t+1}K_{t+1}/N}{Q_t}=\frac{P_{t+1}+D_{t+1}}{P_t} \label{eq:noarbitrage_SL}
\end{equation}
holds. Let
\begin{equation}
    S_t\coloneqq Q_tN+P_tX \label{eq:S}
\end{equation}
be the aggregate asset value and 
\begin{equation}
    E_t\coloneqq r_tK_t+D_tX=F_K(K_t,L_t)K_t+D_t \label{eq:E}
\end{equation}
be the aggregate dividend. Using the no-arbitrage condition \eqref{eq:noarbitrage_SL}, we obtain
\begin{align}
    R_tS_t&=R_t(Q_tN+P_tX) \notag \\
    &=(Q_{t+1}N+r_{t+1}K_{t+1})+(P_{t+1}+D_{t+1})X\notag \\
    &=(Q_{t+1}N+P_{t+1}X)+(r_{t+1}K_{t+1}+D_{t+1}X) \notag \\
    &=S_{t+1}+E_{t+1}. \label{eq:noarbitrage_S}
\end{align}
By the same argument as in the proof of Proposition \ref{prop:textbook}, in equilibrium the aggregate asset value equals aggregate savings: $S_t=\beta w_tL_t=\beta F_L(K_t,L_t)L_t$. We thus obtain the following proposition.

\begin{prop}\label{prop:asset_bubble}
In equilibrium, the aggregate asset value, aggregate dividend, and gross risk-free rate are uniquely given by $S_t=\beta F_L(K_t,L_t)L_t$, \eqref{eq:E}, and \eqref{eq:noarbitrage_S}. There is a bubble in the aggregate asset market if and only if
\begin{equation}
    \sum_{t=1}^\infty \frac{F_K(K_t,L_t)K_t+D_tX}{F_L(K_t,L_t)L_t}<\infty. \label{eq:asset_bubble}
\end{equation}
\end{prop}

\begin{proof}
Immediate from the main text and Lemma \ref{lem:bubble}.
\end{proof}

\subsection{Bubble substitution}

Interestingly, even though the equilibrium allocation is unique, the stock and land prices may be indeterminate. To see why, let $q_0=1$ and $q_t=1/\prod_{s=0}^{t-1}R_s$ be the (unique) Arrow-Debreu prices and define the fundamental values of stock and land by
\begin{align*}
    V_t^S&\coloneqq \frac{1}{q_tN}\sum_{s=t+1}^\infty q_sr_sK_s,\\
    V_t^L&\coloneqq \frac{1}{q_t}\sum_{s=t+1}^\infty q_sD_s.
\end{align*}
Define the aggregate bubble by
\begin{equation*}
    B_t\coloneqq S_t-(V_t^SN+V_t^LX)\ge 0.
\end{equation*}
Using \eqref{eq:S}--\eqref{eq:noarbitrage_S}, we obtain $B_{t+1}=R_tB_t$. For any $\theta\in [0,1]$, define the stock and land prices by
\begin{align*}
    Q_t&=V_t^S+\frac{\theta}{N}B_t,\\
    P_t&=V_t^L+\frac{1-\theta}{X}B_t.
\end{align*}
Then clearly $Q_t,P_t$ satisfy the no-arbitrage condition \eqref{eq:noarbitrage_SL}, so we have a continuum of equilibria indexed by $\theta\in [0,1]$ if there is a bubble ($B_t>0$). It is easy to show that every deterministic equilibrium takes this form.

We note that even though the bubble sizes on individual assets are indeterminate (because stocks and land are perfect substitutes), the total size of the bubble is determinate and hence the consumption allocation is identical regardless of the size of the bubble attached to each asset. This argument is the same as the ``bubble substitution'' argument in \citet[\S5]{Tirole1985}.

\subsection{Productivity growth and bubbles}

Finally, we consider a simple example to study under what conditions bubbles emerge. Let the production function exhibit constant elasticity of substitution (CES), so
\begin{equation}
    F(K,L)=\begin{cases*}
        \left(\alpha K^{1-1/\sigma}+(1-\alpha)L^{1-1/\sigma}\right)^\frac{1}{1-1/\sigma} & if $0<\sigma\neq 1$,\\
        K^\alpha L^{1-\alpha} & if $\sigma=1$,
    \end{cases*} \label{eq:CES}
\end{equation}
where $\sigma>0$ is the elasticity of substitution and $\alpha\in (0,1)$ is a parameter. Suppose capital, labor, and land rent grow at constant rates, so $K_t=K_0G_K^t$, $L_t=L_0G_L^t$, and $D_t=D_0G_X^t$, where $G_K,G_L,G_X>0$. Empirical evidence suggests that the capital-labor substitution elasticity is less than 1,\footnote{See \citet{OberfieldRaval2021} for a study using micro data and \citet*{GechertHavranekIrsovaKolcunova2022} for a literature review and metaanalysis.} so set $\sigma<1$. A straightforward calculation shows
\begin{subequations}
\begin{align}
    F_K(K,L)&=\left(\alpha K^{1-1/\sigma}+(1-\alpha)L^{1-1/\sigma}\right)^\frac{1}{\sigma-1}\alpha K^{-1/\sigma}, \label{eq:straightK}\\
    F_L(K,L)&=\left(\alpha K^{1-1/\sigma}+(1-\alpha)L^{1-1/\sigma}\right)^\frac{1}{\sigma-1}(1-\alpha) L^{-1/\sigma}. \label{eq:straightL}
\end{align}
\end{subequations}
Therefore
\begin{equation}
    \frac{F_K(K,L)K}{F_L(K,L)L}=\frac{\alpha}{1-\alpha}(K/L)^{1-1/\sigma}. \label{eq:straightKL}
\end{equation}
There are two cases to consider.

\begin{case}[$G_K\le G_L$]
In this case, using $\sigma<1$ and \eqref{eq:straightKL}, each term in \eqref{eq:asset_bubble} is positive and bounded away from zero, so the sum in \eqref{eq:asset_bubble} diverges. Therefore there are no bubbles.
\end{case}

\begin{case}[$G_K>G_L$]
In this case, using $\sigma<1$ and \eqref{eq:straightKL}, we have
\begin{equation*}
    \sum_{t=1}^\infty \frac{F_K(K_t,L_t)K_t}{F_L(K_t,L_t)L_t}=\frac{\alpha}{1-\alpha}\sum_{t=1}^\infty (K_0/L_0)^{1-1/\sigma}(G_K/G_L)^{(1-1/\sigma)t}<\infty.
\end{equation*}
Furthermore, using \eqref{eq:straightL} we have
\begin{equation*}
    \frac{D_t}{F_L(K_t,L_t)L_t}\sim \frac{D_0G_X^t}{(1-\alpha)^{1-1/\sigma}L_0G_L^t},
\end{equation*}
whose sum converges if and only if $G_X<G_L$.
\end{case}

Therefore we obtain the following proposition.

\begin{prop}\label{prop:asset_bubble2}
If the production takes the CES form \eqref{eq:CES} with $\sigma<1$ and
\begin{equation*}
    (K_t,L_t,D_t)=(K_0G_K^t,L_0G_L^t,D_0G_X^t),
\end{equation*}
then there is a bubble in the aggregate asset market if and only if $G_K>G_L>G_X$. 
\end{prop}

Proposition \ref{prop:asset_bubble2} implies that bubbles in aggregate stock and land prices necessarily emerge if the capital growth rate exceeds the labor growth rate (possibly due to firm creation and innovation) \footnote{If we consider a situation in which capital is created by entrepreneurs, the capital growth rate means the labor productivity growth rate of entrepreneurs.} and the labor growth rate exceeds the land rent growth rate (possibly due to the declining importance of agriculture). Intuitively, as a bubble is attached to the price of capital with high productivity growth rates, part of earned labor income flows into the real estate sector with lower productivity growth rates, causing land prices to rise at a faster pace than land rents. Because the growth rate of land rent is lower than that of wage income, the rising incomes pull up the land price. This generates land bubbles simultaneously. 

The phenomenon of bubbles in the stock prices of booming sectors or production factors with high productivity growth can be seen in the {U.S.} dot-com bubble in the late 1990s. The result of earned income flowing into the real estate market and causing a bubble can be used to explain housing bubbles. The phenomenon of simultaneous stock and land bubbles may account for the bubbly episodes that many countries experienced, including Japan that experienced substantial increases in land and stock prices in the late 1980s. 

To illustrate the result that simultaneous bubbles emerge in a simple setting, we abstract from realistic elements such as borrowing restrictions or mortgages, which are important for accounting for realistic bubbles. We can expect that incorporating such elements will affect the condition for the simultaneous emergence of stock and land bubbles. In addition, the idea of simultaneous bubbles can be applied to international spillover \ie, how a bubble in one country has a contagious effect on another country, or more generally, how financial market integration affects the emergence of bubbles in each country. In this way, the analysis here has the potential to open up new research directions.

\section{Concluding remarks}

We conclude the chapter by commenting on important implications for macro-theory construction. Throughout this chapter, we emphasized the importance of unbalanced growth. Our construction of macro-finance models that allow unbalanced growth to occur provides a new perspective on the methodology of macro-theory construction. Many macro-finance models are constructed so that the economy converges to a balanced growth path with a constant price-dividend ratio along a single dynamic path that can be drawn with one stroke of brush (usually a saddle path). To think about bubbles attached to real assets, we need to build a model so that a dynamic path that deviates from a stable path converging to balanced growth, \ie, a dynamic path with unbalanced growth, is also possible. 

It should be noted that even if the economy deviates from the balanced growth path and rides on the dynamic path with unbalanced growth, if circumstances unexpectedly change ex post, the economy may return to the balanced growth path where the price-dividend ratio is stable. Looking at this dynamics from an ex post perspective, it appears as if the macro-economy has temporarily left the stable path and taken on a bubble path and then collapsed. During these dynamics, the price-dividend ratio exhibits an explosive increase and fall. 

Moreover, in reality, if policymakers decide that the observed price-dividend ratio appears to be too high, they tend to impose taxes on capital gains or land transactions. If taxes are sufficiently raised, the stock and/or land bubble will surely collapse and the price-dividend ratio will converge to a stable value. With loosening and tightening of the tax policy (in a way contrary to private agents' expectations), the macro-economy may switch back and forth between fundamental and bubbly states, with upward and downward movements in the price-dividend ratio. In reality, this process may repeat itself. Hence, from these reasons, the property of $P_t/D_t \to \infty$ implied by Lemma \ref{lem:bubble} should not be taken literally. 

Furthermore, once we consider aggregate uncertainty, it enriches the dynamics of the price-dividend ratio and provides another new insight. With stochastic and persistent fluctuations in productivity, \citet[\S 4.2]{HiranoToda2025PNAS} show that land prices fluctuate, with the price-dividend ratio rising and falling repeatedly, which appears to be the onset and bursting of a land price bubble. However, land prices always contain bubbles and therefore in an environment with aggregate risks, even if the price-dividend ratio appears to be stable for an extended period of time, it does not necessarily mean land prices reflect fundamentals. So long as the bubble necessity condition with aggregate uncertainty is satisfied, land prices always contain a bubble and the bubble size is changing. When we consider stochastic bubbles, whether they are stock market bubbles or land bubbles, bubbles are expected to collapse at some point in the future. Agents expect the price-dividend ratio to rise explosively as long as the bubble persists, then fall at the time of the burst, and then converge to a stable value at some point in the future. Despite this, bubbles must occur under some conditions. \citet*{HiranoKishiTodaBursting} present such a model and show that the dynamics with stochastic bubbles, which is characterized by unbalanced growth, can be seen as a temporary deviation from a balanced growth path in which asset prices reflect the fundamentals.

\appendix

\section{Proofs}\label{sec:proof}

\subsection{Proof of Proposition \ref{prop:bubble_MM}}

\ref{item:MM1} Under the maintained assumption, we can take $0<\ubar{S}\le \bar{S}$ such that $\ubar{S}\le S_t\le \bar{S}$ for all $t$. Since
\begin{equation*}
    q_tp_t\ubar{S}\le q_tP_t=q_tp_tS_t\le q_tp_t\bar{S},
\end{equation*}
the conditions \eqref{eq:bubble_stock} and \eqref{eq:bubble_value} are equivalent.

\ref{item:MM2} By iterating \eqref{eq:noarbitrage_stock}, we obtain
\begin{equation}
    p_0=\sum_{t=1}^Tq_td_t+q_Tp_T\ge q_Tp_T, \label{eq:qp_ub}
\end{equation}
where we have used $d_t\ge 0$. Similarly, we have
\begin{equation}
    P_0\ge q_TP_T. \label{eq:qP_ub}
\end{equation}
Let $\liminf_{t\to\infty}S_t=0$. Then
\begin{equation*}
    \liminf_{t\to\infty} q_tP_t=\liminf_{t\to\infty} q_tp_tS_t\le \liminf_{t\to\infty} p_0S_t=0,
\end{equation*}
where the inequality follows from \eqref{eq:qp_ub}. Since $q_tP_t$ has a limit, it follows that $q_tP_t\to 0$, so \eqref{eq:bubble_value} implies that there is no bubble in firm value.

\ref{item:MM3} Let $\limsup_{t\to\infty}S_t=\infty$. Suppose to the contrary that there is a bubble in firm stock. Then \eqref{eq:bubble_stock} implies that $q_tp_t$ converges to a positive constant, so by \eqref{eq:qP_ub} we obtain
\begin{equation*}
    P_0\ge \limsup_{t\to\infty} q_tP_t=\left(\lim_{t\to\infty} q_tp_t\right)\left(\limsup_{t\to\infty}S_t\right)=\infty,
\end{equation*}
which is a contradiction. \hfill \qedsymbol

\subsection{Proof of Proposition \ref{prop:wilson}}

Let $P_t>0$ be any equilibrium asset price. Because the old exit the economy, the equilibrium consumption allocation is
\begin{equation*}
    (y_t,z_t)=(aG^t-P_t,bG^t+P_t+D_t).
\end{equation*}
Nonnegativity of consumption implies $P_t\le aG^t$. Let $R_t\coloneqq (P_{t+1}+D_{t+1})/P_t$ be the gross risk-free rate. The first-order condition for optimality together with the nonnegativity of consumption implies that $R_t\ge 1/\beta$, with equality if $P_t<aG^t$. Suppose $R_t>1/\beta$. Then $P_t=aG^t$, so
\begin{align*}
    R_{t-1}\coloneqq \frac{P_t+D_t}{P_{t-1}}&=\frac{aG^t+DG_d^t}{P_{t-1}}\ge \frac{aG^t+DG_d^t}{aG^{t-1}} && (\because P_{t-1}\le aG^{t-1})\\
    &=\frac{aG^{t+1}+DGG_d^t}{aG^t} > \frac{aG^{t+1}+DG_d^{t+1}}{P_t} && (\because G>G_d, P_t=aG^t)\\
    &\ge \frac{P_{t+1}+D_{t+1}}{P_t}=R_t>\frac{1}{\beta}. && (\because P_{t+1}\le aG^{t+1})
\end{align*}
Therefore by induction, if $R_t>1/\beta$, then $R_s>1/\beta$ for all $s\le t$. This argument shows that, in equilibrium, either
\begin{enumerate*}
    \item\label{item:R=} there exists $T>0$ such that $R_t=1/\beta$ for all $t\ge T$, or
    \item\label{item:R>} $R_t>1/\beta$ for all $t$.
\end{enumerate*}
In Case \ref{item:R=}, using $1/R_t=\beta$ for $t\ge T$ and $1/\beta<G_d$, the asset price at time $t\ge T$ can be bounded from below as
\begin{equation*}
    P_t\ge V_t=\sum_{s=1}^\infty \beta^s DG_d^{t+s}=\sum_{s=1}^\infty DG_d^t(\beta G_d)^s=\infty,
\end{equation*}
which is impossible in equilibrium. Therefore it must be Case \ref{item:R>} and hence $P_t=aG^t$ and $y_t=0$ for all $t$. In this case, we have
\begin{equation*}
    R_t=\frac{aG^{t+1}+DG_d^{t+1}}{aG^t}>G>\frac{1}{\beta},
\end{equation*}
so the first-order condition holds and we have an equilibrium, which is unique. Using $P_t=aG^t$, $D_t=DG_d^t$, and applying Lemma \ref{lem:bubble}, we immediately see that there is a bubble. \hfill \qedsymbol

\subsection{Proof of Proposition \ref{prop:bewley}}

Let $g(P)\coloneqq \beta u'(b+P)-u'(a-P)$. Then
\begin{equation*}
    g'(P)=\beta u''(b+P)+u''(a-P)<0,
\end{equation*}
so $g$ is strictly decreasing. Under the maintained assumption, we have $g(0)=\beta u'(b)-u'(a)>0$ and $g(a)=\beta u'(b+a)-u'(0)=-\infty$. By the intermediate value theorem, there exists a unique $P\in (0,a)$ satisfying $g(P)=0$, so \eqref{eq:euler_rich} holds. Using \eqref{eq:euler_rich} and $\beta<1$, we obtain \eqref{eq:euler_poor}.

To show that we have an equilibrium it suffices to show the transversality condition for optimality $\lim_{t\to\infty}\beta^t u'(c_t)P_t=0$, where $c_t$ is the consumption of any agent \citep[p.~237, Example 15.3]{TodaEME}. However, this is obvious because $P_t=P$ is constant, $c_t$ is alternating between two values, and $\beta\in (0,1)$. \hfill \qedsymbol

\subsection{Proof of Lemma \ref{lem:backward}}

We may rewrite \eqref{eq:foc} as
\begin{equation}
    g(P)\coloneqq \frac{U_z}{U_y}(a-P,b'+P'+D')(P'+D')-P=0, \label{eq:foc2}
\end{equation}
where we write $a=a_t$, $b'=b_{t+1}$, etc. Clearly, $g$ is continuous. The monotonicity and quasi-concavity of $U$ (Assumption \ref{asmp:U}) imply that $g$ is strictly decreasing and satisfies $g(0)>0$. The Inada condition $U_y(0,z)=\infty$ implies that $g(a)=0-a<0$. By the intermediate value theorem, there exists a unique $P\in (0,a)$ such that $g(P)=0$. \hfill \qedsymbol

\subsection{Proof of Proposition \ref{prop:f}}

\ref{item:f1} The twice differentiability and strict quasi-concavity of $U$ implies that $(U_y/U_z)(1,Gw)$ is continuous and strictly increasing in $w$. Furthermore, the Inada condition in Assumption \ref{asmp:U} implies that its range is $(0,\infty)$. By the intermediate value theorem, there exists a unique $w_f^*>0$ with $(U_y/U_z)(1,Gw_f^*)=G_d$.

\medskip
\noindent
\ref{item:f2} By \eqref{eq:cond_f}, \ref{item:f1}, and using the homogeneity of $U$, there exists a steady state $\xi^*$ of $\Phi$ if and only if $b/a>w_f^*$. As discussed in the main text, the eigenvalues of $D\phi$ are $\lambda_1>1$ and $\lambda_2\in (0,1)$. Fix large enough $t_0\in \N$. The number of stable eigenvalues (those with $\abs{\lambda}<1$) is 1, which is equal to the number of exogenous initial conditions (which is $\xi_{2t_0}=(G_d/G)^{t_0}$). The number of unstable eigenvalues (those with $\abs{\lambda}>1$) is 1, which is equal to the number of endogenous initial conditions (which is $\xi_{1t_0}=p_{t_0}$). Therefore, by the local stable manifold theorem \citep[p.~111, Theorem 8.9]{TodaEME}, there exists an open neighborhood $\Xi$ of $\xi^*$ with $\phi(\Xi)\subset \Xi$ such that, for any $\xi_{2t_0}$ sufficiently close to $\xi_2^*=0$, there exists a unique $p_{t_0}$ such that $\xi_{t_0}\in \Xi$ and $\xi_t\coloneqq \phi^{t-t_0}(\xi_{t_0})\to \xi^*$. Since $\xi_{2t_0}=(G_d/G)^{t_0}$ and $G_d\in (0,G)$, we can achieve this condition taking $t_0\in \N$ large enough. Let $\xi_t=(\xi_{1t},\xi_{2t})$. Since $\xi_1^*>0$, without loss of generality we may assume $\Xi\subset \R_{++}\times \R$. Then $\xi_t\in \Xi$ implies $p_t=\xi_{1t}>0$. By construction, $\xi_{2t}=(G_d/G)^t>0$. We can then uniquely extend $\set{\xi_t}$ backwards in time by Lemma \ref{lem:backward}.

\medskip
\noindent
\ref{item:f3} By the definition of $\xi_t$, we have $P_t=\xi_{1t}^*G_d^t$. Since $\xi_{1t}^*\to \xi_1^*$, the order of magnitude follows from the characterization of the steady state \eqref{eq:xi_1f}. Since both $P_t$ and dividends grow at rate $G_d$, we have $\sum_{t=1}^\infty D_t/P_t=\infty$, so there is no asset price bubble by Lemma \ref{lem:bubble}. \hfill \qedsymbol

\subsection{Proof of Proposition \ref{prop:b}}
\ref{item:b1} The existence and uniqueness of $w_b^*$ follows from the same argument as in the proof of Proposition \ref{prop:f}, and $w_b^*>w_f^*$ follows from the monotonicity of $(U_y/U_z)(1,Gw)$ and $G_d<G$.

\medskip
\noindent
\ref{item:b2} Using the homogeneity of $U$ and the definition of $w_b^*$, the steady state condition \eqref{eq:xi_1b} is equivalent to
\begin{equation*}
    \frac{b+\xi_1^*}{a-\xi_1^*}=w_b^*\iff \xi_1^*=\frac{w_b^*a-b}{1+w_b^*}.
\end{equation*}
Therefore such $\xi_1^*>0$ exists if and only if $b/a<w_b^*$. The existence of a path $\set{\xi_t^*}_{t=0}^\infty$ by the same argument as in the proof of Proposition \ref{prop:f} if the local stable manifold theorem is applicable, which is the case if $d\neq 0,\pm n$. However, using the steady state condition \eqref{eq:xi_1b}, we obtain
\begin{equation*}
    n-d=\xi_1^*(-U_{yy}+2GU_{yz}-G^2U_{zz})=-\xi_1^*\begin{bmatrix}
        1 & -G
    \end{bmatrix}
    \begin{bmatrix}
        U_{yy} & U_{yz}\\
        U_{yz} & U_{zz}
    \end{bmatrix}
    \begin{bmatrix}
        1 \\ -G
    \end{bmatrix}>0
\end{equation*}
by the strict concavity of $U$. Therefore the case $d=n$ never occurs. If $d>0$, then since $n>d>0$, we have $\lambda_1=n/d>1$, so the path is unique.

\medskip
\noindent
\ref{item:b3} By the definition of $\xi_t$, we have $P_t=\xi_{1t}^*G^t$. Since $\xi_{1t}^*\to \xi_1^*$, the claim follows from the characterization of $\xi_1^*$. Since $P_t$ grows at rate $G$ and dividends grow at rate $G_d$, we have $\sum_{t=1}^\infty D_t/P_t<\infty$, so there is an asset price bubble by Lemma \ref{lem:bubble}. \hfill \qedsymbol

\subsection{Proof of Theorem \ref{thm:lr}}

\ref{item:lr1} Immediate from Proposition \ref{prop:f}.

\medskip
\noindent
\ref{item:lr2} The existence of bubbly equilibria follows from Proposition \ref{prop:b}. The nonexistence of fundamental equilibria when $w<w_f^*$ follows from Theorem 2 of \citet{HiranoToda2025JPE}.

\medskip
\noindent
\ref{item:lr3} The condition $w_f^*<w<w_b^*$ implies that the fundamental steady state $\xi_f^*=(0,0)$ of \eqref{eq:Phi_b} is stable because $\lambda_1,\lambda_2\in (0,1)$. Therefore there exist a continuum of equilibria such that $G^{-t}P_t\to 0$. In any such equilibrium, using \eqref{eq:Rt}, Assumption \ref{asmp:G}, and the homogeneity of $U$, the gross risk-free rate becomes
\begin{align*}
    R_t&=\frac{U_y}{U_z}(y_t,z_{t+1})=\frac{U_y}{U_z}(a_t-P_t,b_{t+1}+P_{t+1}+D_{t+1})\\
    &=\frac{U_y}{U_z}(1-(1/a)(P_t/G^t),Gw+(G/a)(P_{t+1}/G^{t+1})+(DG_d/a)(G_d/G)^t)\\
    &\to \frac{U_y}{U_z}(1,Gw)>G_d
\end{align*}
as $t\to\infty$. Therefore letting $R\coloneqq (U_y/U_z)(1,Gw)>G_d$ and $V_t$ be the fundamental value of the asset, we have
\begin{equation*}
    \lim_{t\to\infty} V_t/G_d^t=\frac{DG_d}{R-G_d}>0.
\end{equation*}
Thus, in any fundamental equilibrium, $G_d^{-t}P_t$ converges to this positive number. However, Proposition \ref{prop:f} shows that such an equilibrium is unique. Therefore all equilibria except this one are bubbly. \hfill \qedsymbol

\printbibliography

\end{document}